\pdfoutput=1
\documentclass[%
 reprint,
superscriptaddress,
 amsmath,amssymb,
 aps,
 prl,
]{revtex4-2}

\usepackage{graphicx}
\usepackage{dcolumn}
\usepackage{bm}
\usepackage{hyperref}
\usepackage{color}
\usepackage{amsthm}

\newcommand{\name}{{iQCT}}

\newcommand{\Tr}[0]{\mathrm{Tr}}

\newtheorem{theorem}{Theorem}

\begin{document}

\preprint{APS/123-QED}

\title{Quantum Network Tomography via Learning Isometries on Stiefel Manifold}

\author{Ze-Tong Li}
\affiliation{State Key Laboratory of Millimeter Waves, Southeast University, Nanjing 210096, China.}
\affiliation{%
Frontiers Science Center for Mobile Information Communication and Security, Southeast University, Nanjing 210096, China.
}%
\affiliation{%
Purple Mountain Lab, Nanjing 211111, China.}%

\author{Xin-Lin He}%
\author{Cong-Cong Zheng}%
\affiliation{State Key Laboratory of Millimeter Waves, Southeast University, Nanjing 210096, China.}
\affiliation{%
Frontiers Science Center for Mobile Information Communication and Security, Southeast University, Nanjing 210096, China.
}%

\author{Yu-Qian Dong}
\affiliation{Yangtze Delta Region Industrial Innovation Center of Quantum and Information Technology, Suzhou 215100, China.}

\author{Tian Luan}
\affiliation{Yangtze Delta Region Industrial Innovation Center of Quantum and Information Technology, Suzhou 215100, China.}

\author{Xu-Tao Yu}
\email{yuxutao@seu.edu.cn}
\affiliation{State Key Laboratory of Millimeter Waves, Southeast University, Nanjing 210096, China.}
\affiliation{%
Frontiers Science Center for Mobile Information Communication and Security, Southeast University, Nanjing 210096, China.
}%
\affiliation{%
Purple Mountain Lab, Nanjing 211111, China.}%

\author{Zai-Chen Zhang}
\email{zczhang@seu.edu.cn}
\affiliation{%
National Mobile Communications Research Laboratory, Southeast University, Nanjing 210096, China.
}%
\affiliation{%
Frontiers Science Center for Mobile Information Communication and Security, Southeast University, Nanjing 210096, China.
}%
\affiliation{%
Purple Mountain Lab, Nanjing 211111, China.}%

\date{\today}

\begin{abstract}
  Explicit mathematical reconstructions of quantum networks play a significant role in developing quantum information science. However, tremendous parameter requirements and physical constraint implementations have become computationally non-ignorable encumbrances. In this work, we propose an efficient method for quantum network tomography by learning isometries on the Stiefel manifold. Tasks of reconstructing quantum networks are tackled by solving a series of unconstrained optimization problems with significantly fewer parameters. The stepwise isometry estimation shows the capability for providing information of the truncated quantum network while processing the tomography. Remarkably, this method enables the dimension-reduced quantum network tomography by reducing the ancillary dimensions of isometries with bounded error. As a result, our proposed method exhibits high accuracy and efficiency.
\end{abstract}

\maketitle


\section{\label{sec:introduction}Introduction}
Quantum networks are extremely important in quantum communication, computation, metrology \cite{wei2022Real,maniscalco2007Entanglement,caruso2014Quantum,kos2023Circuits,wolf2008Assessing,bylicka2014NonMarkovianity,paulson2021Hierarchy,zhao2022Quantum,sarovar2020Detecting,wei2022Hamiltonian,muller2019Understanding,wilen2021Correlated,dial2016Bulk,brownnutt2015Iontrap}. A quantum network, which can be considered as a multi-time-step combination of elementary quantum circuits \cite{chiribella2009Theoretical}, has capabilities of performing complex tasks that require multiple input-output states at different time steps, as a non-Markovian quantum process \cite{gutoski2007General}. Furthermore, the quantum network is competent to model non-Markovian quantum noise \cite{wei2022Hamiltonian,muller2019Understanding,wilen2021Correlated,dial2016Bulk,brownnutt2015Iontrap,parrado-rodriguez2021Crosstalk,kuhlmann2013Charge,yoneda2023Noisecorrelation,rojas-arias2023Spatial} resulting from indispensable system-environment correlations, and promotes development of clean quantum computers \cite{bonillaataides2021XZZX,tuckett2018Ultrahigh,harper2023Learning,nautrup2019Optimizing,farrelly2021Tensornetwork,suzuki2022Quantum,wang2023DGR}.

Explicit mathematical reconstructions of quantum networks play a significant role in the development of quantum information science. For example, quantum computer manufacturers can leverage the information of non-Markovian quantum noise to enhance the reliability of quantum computers. A prevalent way to model a quantum network is the quantum comb \cite{chiribella2009Theoretical,pollock2018NonMarkovian,milz2021Quantum}. As shown in Fig.~\ref{fig:quantum_comb}, an $N$-time-step quantum comb constructs a completely positive (CP) map from $N$ input states to $N$ output states, labeled by even and odd numbers, respectively, with causality that later input systems cannot influence previous output systems. Then, the quantum comb can be represented by a matrix with massive entries, such as the Choi-Jamiołkowski isomorphism (CJI), a.k.a. Choi states, which is well established in quantum-information theory \cite{poyatos1997Complete}. Several methods for quantum comb tomography (QCT) have been proposed and exhibited outstanding effectiveness. \cite{chiribella2009Theoretical,milz2018Reconstructing,pollock2018NonMarkovian,milz2021Quantum,white2022NonMarkovian,white2023Filtering,white2020Demonstration,white2021Manybody}. 

\begin{figure}[t]
  \includegraphics[width=0.45\textwidth]{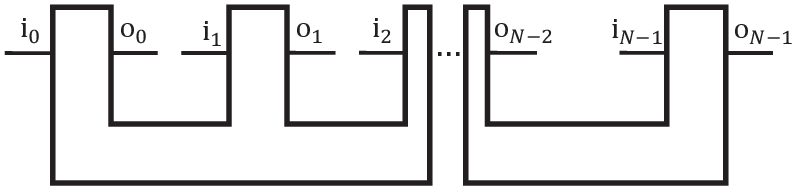}
  \caption{\label{fig:quantum_comb} A quantum comb with $N$ time steps. Wires labeled by $\mathtt{i}_k$ and $\mathtt{o}_k$ represent the input and output systems, respectively, at time step $k$, $k=0,1,\dots,N-1$. Causality of the quantum comb indicates that the information flows along with the time step, which implies that the input system at time step $l$ cannot influence the output system $m$ if $l>m$.}
\end{figure}

However, there still exist two challenges in performing QCT. First, recent QCT methods require multitudinous parameters to completely represent an arbitrary quantum network, that is, the productions of squared dimensions of all input-output systems. The tremendous and increasing parameter requirement w.r.t. the time step and dimensions of quantum states becomes a computationally non-ignorable encumbrance. Second, the physical properties of a quantum comb requires implementations of CP and causal (CPC) constraints in QCT. This leads to introducing massive equality constraints to a positive semidefinite complex matrix \cite{white2022NonMarkovian}.

These two challenges result in the QCT being intractable with increasing time steps and system dimensions. Efficient methods with capabilities of tackling massive parameters and constraints are urgently required, which directly motivates this work.

In this work, we propose an efficient isometry-based QCT (\name) on the Stiefel manifold. The {\name} equivalently implements the low-rank estimation to the process tensor of the quantum comb by a list of isometries with significantly fewer parameters and bounded error. The original QCT task is transformed into solving $N$ unconstrained optimization problems on the Stiefel manifold \cite{boumal2023Introduction,ahmed2023Gradientdescent,rapcsak2002Minimization,li2020Efficient}. Hence, the CPC properties are inherently satisfied without introducing any constraint. The stepwise optimization determines one isometry at each time step. Hence, our method can provide information about the truncated quantum comb while processing the tomography. As a result, our proposed method exhibits high accuracy and efficiency. Remarkably, the {\name} especially suitable for the cases when the quantum comb has mild time correlations such that the purity of the Choi state is high, or the dimension of the environment is low and known to the experimentor.

\section{Framework of Isometry based QCT}

A quantum comb $\mathcal{C}^{(N)}$ with $N$ time steps as shown in Fig.~\ref{fig:quantum_comb}, that represents an $N$-time-step quantum network, maps $N$ input systems from $\mathrm{Lin}(\mathcal{H}_{\mathtt{i}_k})$ to $N$ output systems from $\mathrm{Lin}(\mathcal{H}_{\mathtt{o}_k})$, $k=0,1,\dots,N-1$, where $\mathcal{H}_x$ is a $d_x$-dimensional Hilbert space and $\mathrm{Lin}(\mathcal{H}_{x})$ is the space of linear operator on $\mathcal{H}_{x}$. Let $\mathcal{H}_\mathrm{in}^{(N)} := \bigotimes_{k=0}^{N-1}\mathcal{H}_{\mathtt{i}_k}$ and $\mathcal{H}_\mathrm{out}^{(N)}  := \bigotimes_{k=0}^{N-1}\mathcal{H}_{\mathtt{o}_k}$. 

Based on the Stinespring dilation theorem \cite{gutoski2007General,chiribella2009Theoretical}, the quantum comb $\mathcal{C}^{(N)}$ can be implemented by a sequence of isometries 
\begin{align}\label{eq:comb_isometry_repr}
  \mathcal{C}^{(N)}(\rho) = \mathrm{Tr}_{A_N}\![V^{(N-1)}\dots V^{(0)}\rho V^{(0)\dagger}\dots V^{(N-1)\dagger}],
\end{align}
where $\rho \in \mathrm{Lin}(\mathcal{H}_\mathrm{in})$, $\mathcal{C}^{(N)}(\rho) \in \mathrm{Lin}(\mathcal{H}_\mathrm{out})$, $V^{(k)}$ is an isometry from $\mathcal{H}_{\mathtt{i}_k}\otimes\mathcal{H}_{A_k}$ to $\mathcal{H}_{\mathtt{o}_k}\otimes\mathcal{H}_{A_{k+1}}$, and $\mathcal{H}_{A_k}$ is the $d_{A_k}$ dimensional Hilbert ancillary space of the output of isometry $V^{(k-1)}$. Note that we omit the identity operators on the Hilbert spaces in \eqref{eq:comb_isometry_repr} henceforth.

Then, a quantum comb can be reconstructed with bounded Hilbert-Schmidt (HS) distance, as shown in Theorem~\ref{thm:bound}:
\begin{theorem}\label{thm:bound}
  Let $\Upsilon$ denote the Choi state of an $N$-time-step quantum comb $\mathcal{C}^{(N)}$. Suppose $\mathcal{P}=\mathrm{Tr}({\Upsilon}^2)/(\Tr\Upsilon)^2$ is the purity of normalized $\Upsilon$, there exists a list of isometries $\mathbb{V}:=[V^{(t)}:= \mathcal{H}_{\mathtt{i}_{t}}\otimes\mathcal{H}_{A_t} \to \mathcal{H}_{\mathtt{o}_t}\otimes\mathcal{H}_{A_{t+1}}]_{t=0}^{N-1}$ with $O(\max_{0\leq t< N}(d_{\mathtt{i}_{t}}d_{\mathtt{o_{t}}}) R^2)$ entries that implements the Choi state $\Upsilon^\prime$, such that the Hilbert-Schmidt distance
  \begin{align}
    D_{HS} \leq (\Tr\Upsilon)^2[\mathcal{P}-R\beta_{+}^2(K) + \frac{1}{R^2}\left( 1-R \beta_{+}(K)\right)^2],
  \end{align}
  when $1 \leq R < K-1$, and 
  \begin{align}
    D_{HS} \leq \max_{l\ge R} (\Tr\Upsilon)^2(l-R)\left[1+\frac{l-R}{R^2}\right] \beta^2_{-}(l),
  \end{align}
  when $K \leq R < d_{\rho}$, where $D_{HS}(\Upsilon^\prime,\Upsilon) = \Tr[(\Upsilon^\prime-\Upsilon)^2]$, $\beta_{\pm}(x)=\frac{1}{x} \pm \sqrt{ \frac{1}{x(x-1)}\left( \mathcal{P}-\frac{1}{x} \right) }$, $K = \lceil \frac{1}{P} \rceil$.
\end{theorem}
The proof of Theorem~\ref{thm:bound} is refered to the supplementary material. This provides a suggestion of ancillary dimensions that $\max_{k=1,\dots,N}{d_{A_k}} = R$ by assessing the maximal Hilbert-Schmidt distance with given purity of the Choi state. Note that the estimation of the purity of the Choi state can be efficiently tackled by classic shadow, and even does not require high accuracy. Remarkably, Theorem~\ref{thm:bound} also indicates that the quantum comb $\mathcal{C}^{(N)}$ can be reconstructed by significantly less complex parameters with low Hilber-Schmidt distance, when the non-Markovianity is mild such that the purity of the Choi state is high, such as the cross-talk introduced by the imperfect signal line to control the double-qubit gate between the system and an idle qubit on the chip.

We assume that the target quantum network is invariable during the experiments and can be accessed multiply. To perform QCT on the quantum network, the experimenter prepares known tomographically complete state sets $\Gamma^{(k)}:=\{\rho^{(k)}_{i}\}$, $i=0,\dots,d_{\mathtt{i}_k}^2-1$, and measurement sets $\Xi^{(k)}:=\{E^{(k)}_{j}\}$, $j=0,\dots,d_{\mathtt{o}_k}^2-1$, that span $\mathrm{Lin}(\mathcal{H}_{\mathtt{i}_k})$ and $\mathrm{Lin}(\mathcal{H}_{\mathtt{o}_k})$, respectively, for $k=0,1,\dots,N-1$.

\begin{figure}[t]
  \includegraphics[width=0.48\textwidth]{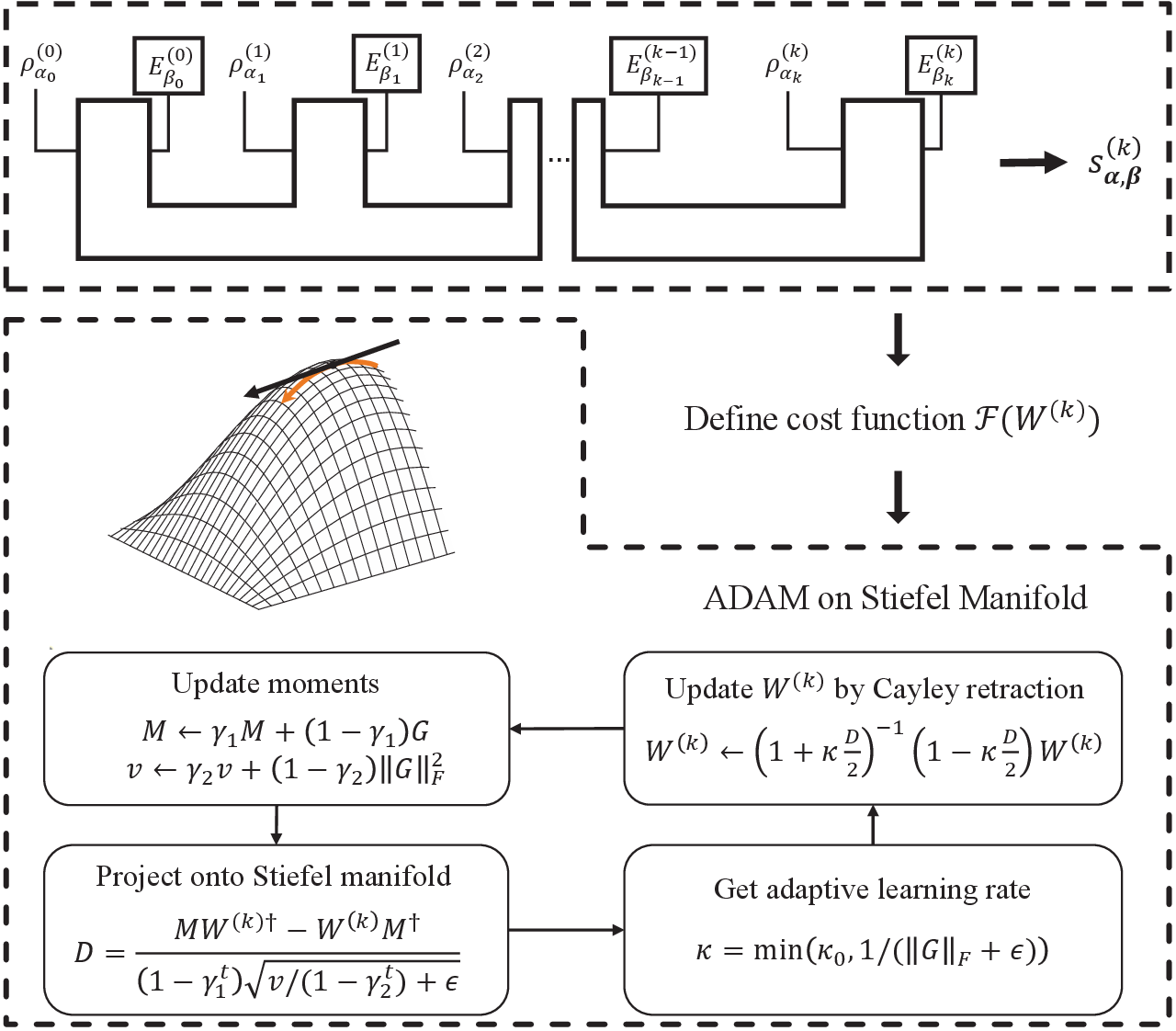}
  \caption{\label{fig:work_flow} Workflow for estimating $V^{(k)}$ in the target $N$-time-step quantum network. The experimenter tests the truncated quantum network with time steps $t=0,\dots,k$, $k\le N-1$, from the target quantum network, as shown at the top of the figure. Then, the cost function is defined by measurement results $s_{\bm{\alpha},\bm{\beta}}^{(k)}$ w.r.t. $W^{(k)}$. Finally, we use the ADAM on the Stiefel manifold to determine $V^{(k)} = \arg\min_{W^{(k)}\in \mathrm{St}^{(k)}}\mathcal{F}(W^{(k)})$.}
\end{figure}

The causality indicates that the input systems at later time steps cannot influence previous output systems, which enables the stepwise optimization in the \name. See supplementary material for details. At time step $k$, only the isometry $V^{(k)}$ is reconstructed with the known $V^{(t)}$, $t<k$. The workflow for estimating $V^{(k)}$ is summarized in Fig.~\ref{fig:work_flow}. The experimenter conducts experiments that combine the input states $\{\rho_{\alpha_0}^{(0)},\dots,\rho_{\alpha_k}^{(k)}\}$ and measurements on output states $\{E_{\beta_0}^{(0)},\dots,E_{\beta_k}^{(k)}\}$ and records the results $s^{(k)}_{\bm{\alpha},\bm{\beta}}$, where $\bm{\alpha} := [\alpha_0, \dots, \alpha_k]$, $\bm{\beta} := [\beta_0, \dots, \beta_k]$. The criteria for selecting $\bm{\alpha}$ and $\bm{\beta}$ are that $\{\eta^{(k-1)}_{\bm{\alpha},\bm{\beta}}\}$ consists of at least $d^2_{\mathtt{i}_k}d^2_{A_k}$ linear independent matrices and that $\beta_k$ spans $\{0,\dots,d^2_{\mathtt{o}_k} -1\}$, where
\begin{gather}
  \eta^{(t)}_{\bm{\alpha},\bm{\beta}} = \Tr_{\mathtt{o}_t} [\rho^{(t+1)}_{\alpha_{t+1}} E^{{(t)}}_{\beta_t}V^{(t)}\eta^{(t-1)}_{\bm{\alpha},\bm{\beta}}V^{(t)\dagger}], t \ge 0,
\end{gather}
and $\eta^{(-1)}_{\bm{\alpha},\bm{\beta}} = \rho^{(0)}_{\alpha_0}$. From \eqref{eq:comb_isometry_repr}, the recovered probability is 
\begin{align}\label{eq:prob_isometry_k}
  p_{\bm{\alpha},\bm{\beta}}(W^{(k)}) = \Tr[E^{(k)}_{\beta_{k}}W^{(k)} \eta^{(k-1)}_{\bm{\alpha},\bm{\beta}} W^{(k)\dagger}].
\end{align}
Then, the isometry $V^{(k)}$ is reconstructed by optimizing the cost function $\mathcal{F}$ on the Stiefel manifold without constraints
\begin{align}\label{eq:stiefel_problem}
  \min_{W^{(k)}\in\mathrm{St}^{(k)}}\mathcal{F}(W^{(k)})= \sum_{\bm{\alpha},\bm{\beta}} \vert\tilde{p}_{\bm{\alpha},\bm{\beta}} - p_{\bm{\alpha},\bm{\beta}}(W^{(k)})\vert^2,
\end{align}
where $\tilde{p}_{\bm{\alpha},\bm{\beta}}=s^{(k)}_{\bm{\alpha},\bm{\beta}}/n_s$ represents the measurement probability, $n_s$ is the total number of samples, $\mathrm{St}^{(k)}:=\{X\in \mathbb{C}^{(k)}:X^\dagger X = I\}$ represents the Stiefel manifold on which $V^{(k)}$ lies, and $\mathbb{C}^{(k)}:=\mathbb{C}^{d_{\mathtt{o}_k}d_{A_{k+1}}\times d_{\mathtt{i}_k}d_{A_k}}$.

Note that the stepwise optimization determines one isometry at each time step. Hence, the iQCT has the capability of providing isometries of $C^{(k)}$, $k\le N$, while performing tomography to $C^{(N)}$. The causality indicates that the isometries of $C^{(k)}$, $k\le N$, completely characterize the truncated quantum network from time step $0$ to $k-1$. This property facilitates experimenters to analyze the currently determined information of the truncated quantum network when the iQCT is determining later isometries.

Remarkably, the proposed method is both efficient and effective when the experimenter has prior information about the required ancillary dimensions beyond the purity of Choi state of the quantum comb. For example, it is reasonable to set the ancillary dimensions as $2^{n_{\max}}$ when prior information suggests that there are at most $n_{\max}$ qubits, excluding the input and output systems, in the quantum network. The ancillary dimensions can also be reduced when characterizing non-Markovian quantum noise with mild system-environment correlations, for example, the cross talk between the quantum system and few environment qubits.

Note that this method can be implemented by instruments without loss of efficiency, and has the capability to characterize non-Markovian quantum noise on the recent noisy intermediate-scale quantum devices. The modification to the instrument-based {\name} is described in the supplementary material.

\section{Reconstruct Isometry by ADAM on Stiefel Manifold}\label{sec:stiefel_adam}

To solve the optimization problem defined in \eqref{eq:stiefel_problem}, we introduce the adaptive moment estimation (ADAM) method \cite{kingma2014Adam,li2020Efficient,reddi2019Convergence,zou2019Sufficient} on the Stiefel manifold. See the supplementary material for details. The Stiefel ADAM method iteratively optimizes the cost function that preserves the orthonormality of parameter matrix. At iteration $t$, the Stiefel ADAM updates the biased first and second moments for the descent direction as
\begin{align}
  M &\leftarrow \gamma_1 M - (1 - \gamma_1) G,\\
  v &\leftarrow \gamma_2 v + (1 - \gamma_2)\|G\|_{F}^2,
\end{align}
respectively, where $G=\partial \mathcal{F}(W^{(k)})/\partial W^{(k)*}$, $\gamma_1$ and $\gamma_2$ are previously selected constants, and $\|\cdot\|_{F}$ represents the Frobenius norm. Then, based on the canonical innerproduct, the moments are projected onto the tangent space of the Stiefel manifold with bias correction 
\begin{align}
  D = \frac{MW^{(k)\dagger} - W^{(k)}M^{\dagger}}{(1-\gamma_1^t)\sqrt{v/(1-\gamma_2^t)+\epsilon}},
\end{align}
where $\epsilon>0$ is a small constant. Finally, the isometry $W^{(k)}$ is updated by Cayley transformation
\begin{align}
  W^{(k)} \leftarrow (I+\kappa\frac{D}{2})^{-1}(I - \kappa\frac{D}{2})W^{(k)},
\end{align}
where $\kappa = \min( \kappa_0, 1 / (\|D\|_F + \epsilon))$ is the adaptive step size, and $\kappa_0$ is the pre-selected maximum step size. The optimization loop is terminated when $\|G W^{(k)\dagger} - W^{(k)}G^\dagger\|_F<\delta$, where $G W^{(k)\dagger}W^{(k)} - W^{(k)}G^\dagger W^{(k)} = \nabla_{W^{(k)}\in\mathrm{St}^{(k)}}\mathcal{F}(W^{(k)})$ is the Riemannian gradient of $\mathcal{F}$ on the Stiefel manifold, and $\delta>0$ is a hyperparameter.


\section{Experimental Results}

To showcase the performance of our proposed method, we first conduct simulations that applying the iQCT to a series of random quantum networks defined by isometries. These numerical simulations are conducted by C++ on the computer enpowered by 2 $\times$ AMD EPYC 7742 CPUs with 1 TB RAM. Measurement probabilities $\tilde{p}_{\bm{\alpha},\bm{\beta}}$ are idealy determined to compose the cost function. Hence, we can benchmark our method in noiseless circumstances to demonstrate the uppder bound of performances. Main criteria of accuracy and efficiency are the Hilbert-Schimdt distance $D_{HS}(\Upsilon)$ between the reconstructed Choi state $\Upsilon^\prime$ and target Choi state $\Upsilon$ and the absolute running time $\Delta_T$ (ms) running on the same computer to fairly demonstrate the efficiency. Code is available in Ref.~\cite{2024Code}.

We first consider the cases that are well-suited for the iQCT, where the ancillary dimensions of the target isometries are far smaller than the maximal ancillary dimensions. Specifically, we set the ancillary demensions as $d_{A_{k}}=\prod_{t=0}^{k-1}d_{\mathtt{i}_{t}}$, while the full ancillary dimensions are $d_{A_{k}}=\prod_{t=0}^{k-1}d_{\mathtt{i}_t}d_{\mathtt{o}_t}$. Meanwhile, the ancillary dimensions are known to the experimenter.

In Fig.~\ref{fig:analytic}(a) and (b), we show the HS distance and absolute running time w.r.t. dimensions of input and output states. We set the dimensions of input and output states at the same time step to be identical. Labels `$n$-$m$' represent that states at time step $0$ and $1$ consist of $n$ and $m$ qubits, respectively. All optimizations are terminated under the same condition by setting $\delta=10^{-4}$. From the results, the iQCT efficiently reconstructs quantum networks with low HS distances. The low variance of HS distances and absolute running times under the same settings indicates the stability of the iQCT in ideal circumstances.

\begin{figure}[t]
  \includegraphics[width=0.45\textwidth]{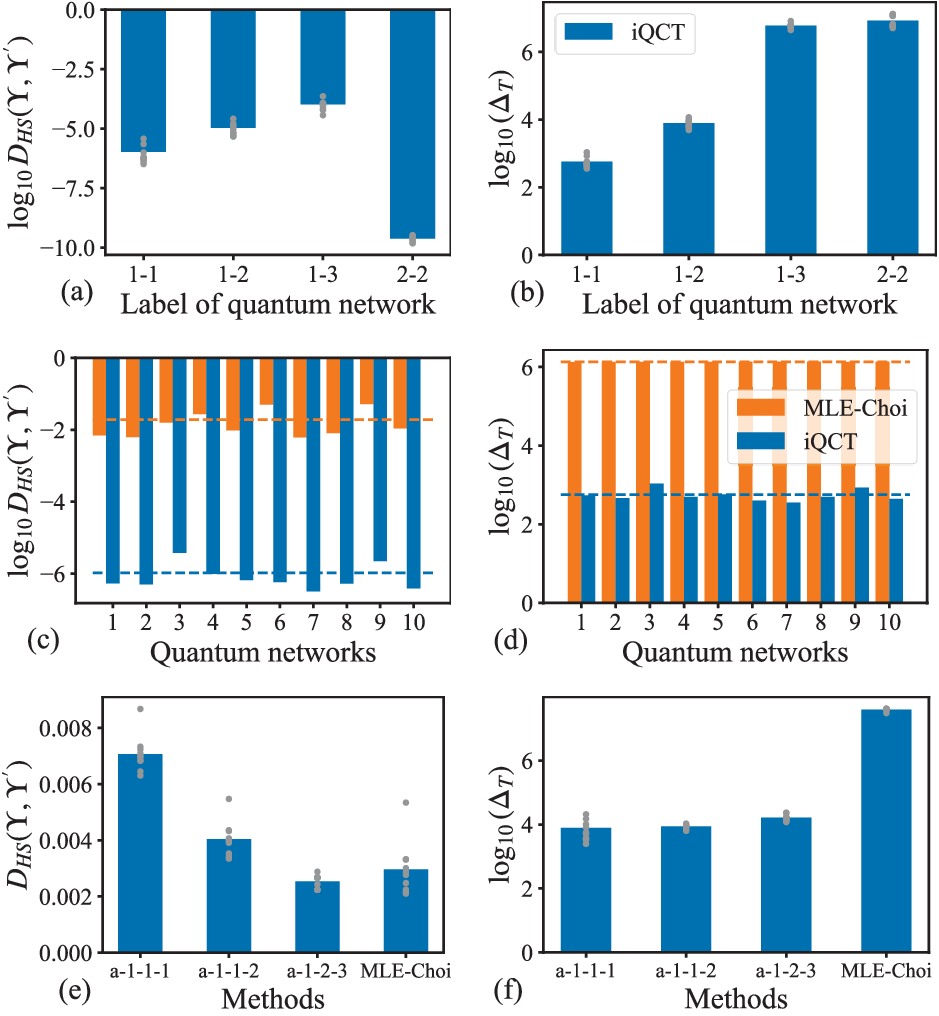}
  \caption{\label{fig:analytic} Results of HS distance and absolute running time. Bars and gray points in (a), (b), (e), and (f) represent average values and values of individual results of the random quantum networks, respectively. (a) and (b) are the results of QCT for 10 random `$n$-$m$' quantum networks. (c) and (d) are the results of QCT for 10 random `1-1' quantum networks, where dashed lines and bars represent average values and values of individual results of networks, respectively. (e) and (f) are results of reconstructing non-Markovian quantum noise. We implements iQCT labeled by `a-$\log_2 d_{A_1}$-$\log_2 d_{A_2}$-$\log_2 d_{A_3}$' and the MLE-Choi, respectively.}
\end{figure}

We further compare the iQCT to the recent state-of-the-art QCT method which construct the Choi state based on the maximum likelihood estimation (MLE) with physical constraints implemented by Dikstra projection \cite{white2022NonMarkovian}, labeled by MLE-Choi. 
In Fig.~\ref{fig:analytic}(c) and (d), we show the fidelity and absolute running time of the two schemes in estimating 10 random `1-1' quantum networks. In this situation, our proposed iQCT 99.96\% improvements to the absolute running time when the iQCT obtains 0.0193 reduction (corresponding to the 99,99\% improvements) to the HS distance on average. We also tried to perform MLE-Choi to other instances in Fig.~\ref{fig:analytic}(a) and (b). However, it is hard to determine results within acceptable running times due to extremely large dimensions of Choi states. 

\begin{figure}[t]
  \includegraphics[width=0.42\textwidth]{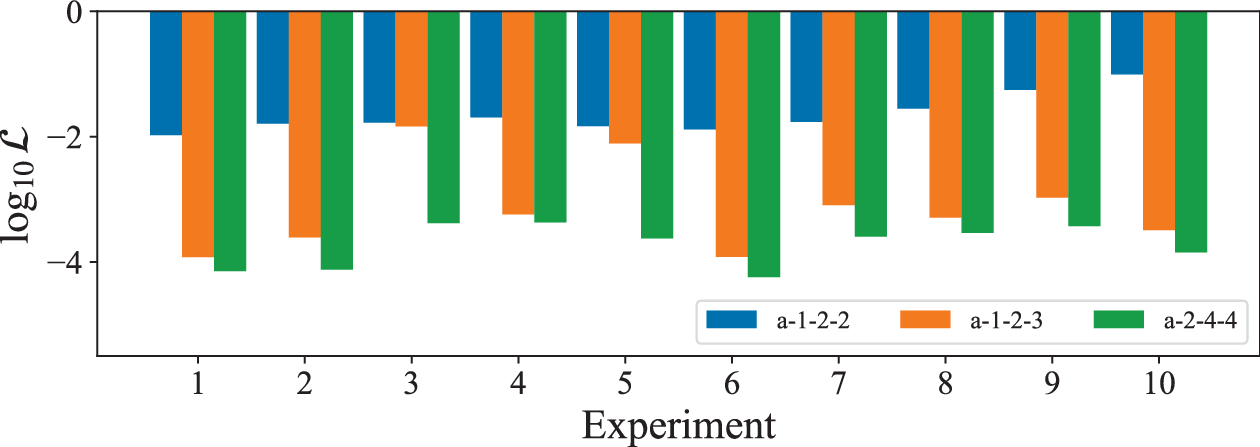}
  \caption{\label{fig:rc_cmp} Results on real quantum chips. Labels in legends represent `a-$\log_2 d_{A_0}$-$\log_2 d_{A_1}$-$\log_2 d_{A_2}$'.}
\end{figure}

Another application of iQCT is operationally characterizing non-Markovian quantum noise in a quantum computer. In the simulation, we consider the single qubit system with 3 time steps. we assume that ZZ cross talk is the main resource of the non-Markovian noise. The system-environment (SE) unitaries are set as $\exp (i \sum_{i,j} a_{P_{S},P_{E}} P_{S} \otimes P_{E})$, where $P_{S}, P_{E}\in\{I,X,Y,Z\}$, $a_{ZZ}=0.05$, and other coefficient are uniformly sampled from 0 to 0.001. Initial environment states are set as $\vert 0 \rangle \langle{0}\vert$. In Fig.~\ref{fig:analytic}(e) and (f), we show the HS distance and the absolute running time in variant configurations. From the result, the iQCT reconstruct the non-Markovian quantum noise with significant lower running times and comparative accuracy w.r.t. the MLE-Choi.

For this application, we also apply iQCT to a single-qubit 3-time-step real quantum device through interventions of instruments. See the supplementary material for details. In this scenario, non-measurement quantum operations are known as unitaries, which means that the tomography results are responsible for the unitary gates instead of CPTP operations. We utilize the relative cost $\mathcal{L}(\mathcal{C}_{\mathrm{dr}},\mathcal{C}_{\mathrm{full}})$ to measure distances between reconstructed dimension-reduced quantum network $\mathcal{C}_{\mathrm{dr}}$ and full-ancillary-dimension quantum network $\mathcal{C}_{\mathrm{full}}$. $\mathcal{L}=0$ when $\mathcal{C}_{\mathrm{dr}}$ is equivalent to $\mathcal{C}_{\mathrm{full}}$ in cases that non-measurement operations span the unitary space at each intermediate time step. From Fig.~\ref{fig:rc_cmp}, iQCT achieves significant similarity between $\mathcal{C}_{\mathrm{dr}}$ and $\mathcal{C}_{\mathrm{full}}$ when the ancillary dimensions are larger than $d_{A_0}=2$, $d_{A_1}=4$, and $d_{A_2}=8$. This indicates the fact that the ancillary dimensions of the real quantum chip are limited, and has potentiality to efficiently reconstruct non-Markovian quantum noise with fewer computational resources but high accuracy.

\begin{figure}[t]
  \includegraphics[width=0.30\textwidth]{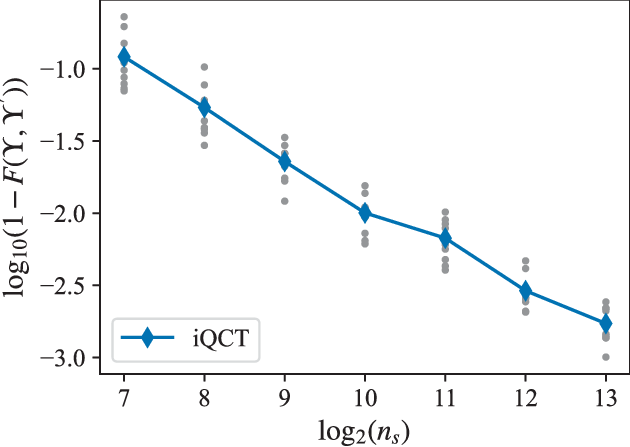}
  \caption{\label{fig:fidelity_samples} Logarithmic Fidelity gaps to $1$ with respect to the number of samples. Colored line represents the average values, while gray points are values of individual results of the random quantum networks}
\end{figure}

Then, we perform the iQCT on 10 quantum networks with 3 time steps, where probabilities are estimated via sampling. Note that we do not adopt any advanced measurement strategies, such as classic shadow \cite{huang2022Learning,jnane2024Quantum}, although we could adopt these strategies by moderately modifying the cost function. In Fig.~\ref{fig:fidelity_samples}, we demonstrate the fidelity w.r.t. the number of samples to showcase robustness against probability estimation errors. The accuracy of iQCT, indicated by the fidelity, shows a linear correlation with the accuracy of probability estimation.

\section{Conclusion}
In this work, we introduce an efficient and powerful isometry-based quantum comb tomography method on the Stiefel manifold, called iQCT, where the isometries are optimized by Stiefel ADAM. Our approach mathematically reconstructs target quantum networks modeled by quantum combs by learning isometries that hold orthonormal constraints on the Stiefel manifold, and thus the CPC constraints are inherently satisfied. The stepwise optimization determines one isometry at each time step. Therefore, iQCT has the capability to provide information about the truncated quantum comb during the tomography process.

The iQCT enables dimension-reduced quantum comb tomography by reduce the ancillary dimensions of isometries with bounded error. The requirement of parameters are significantly reduced compared with the conventional Choi-state-based method. As a result, the iQCT achieves significant improvements on the efficiency with bounded loss of accuracy. Our proposed method is especially suitable for the cases  when the experimenter has the prior information that the time correlations are limited and can be characterized by low-dimensional isometries, or for characterizing non-Markovian quantum noise with mild system-environment correlations.

From experimental results, our proposed method exhibits both high accuracy and efficiency. Particularly, compared with the state-of-the-art Choi-state-based QCT method, iQCT achieves 99.96\% average improvements in efficiency, when the iQCT performs better accuracy, for characterizing small and short quantum networks. These advantages are enhanced when the dimensions of input and output states and the number of time steps increase. Moreover, experiments on the real quantum device indicate the potentiality to efficiently reconstruct non-Markovian quantum noise by iQCT with fewer  computational resources but low HS distances.

\begin{acknowledgments}
  This work is supported by the National Natural Science Foundation of China No.62471126, No.61960206005, and No.61871111, Fundamental Research Funds for the Central Universities 2242022k60001, Jiangsu Key R\&D Program Project BE2023011-2.
\end{acknowledgments}

\bibliography{citations}

\end{document}


\title{Supplementary Material to \\ "Quantum Network Tomography via Learning Isometries on Stiefel Manifold"}

\maketitle

\section*{Proof and Demonstration to Theorem 1}
\begin{theorem}\label{thm:bound}
  Let $\Upsilon$ denote the Choi state of an $N$-time-step quantum comb $\mathcal{C}^{(N)}$. Suppose $\mathcal{P}=\mathrm{Tr}({\Upsilon}^2)/(\Tr\Upsilon)^2$ is the purity of normalized $\Upsilon$, there exists a list of isometries $\mathbb{V}:=[V^{(t)}:= \mathcal{H}_{\mathtt{i}_{t}}\otimes\mathcal{H}_{A_t} \to \mathcal{H}_{\mathtt{o}_t}\otimes\mathcal{H}_{A_{t+1}}]_{t=0}^{N-1}$ with $O(\max_{0\leq t< N}(d_{\mathtt{i}_{t}}d_{\mathtt{o_{t}}}) R^2)$ entries that implements the Choi state $\Upsilon^\prime$, such that the Hilbert-Schmidt distance
  \begin{align}
    D_{HS} \leq (\Tr\Upsilon)^2[\mathcal{P}-R\beta_{+}^2(K) + \frac{1}{R^2}\left( 1-R \beta_{+}(K)\right)^2],
  \end{align}
  when $1 \leq R < K-1$, and 
  \begin{align}
    D_{HS} \leq \max_{l\ge R} (\Tr\Upsilon)^2(l-R)\left[1+\frac{l-R}{R^2}\right] \beta^2_{-}(l),
  \end{align}
  when $K \leq R < d_{\rho}$, where $D_{HS}(\Upsilon^\prime,\Upsilon) = \Tr[(\Upsilon^\prime-\Upsilon)^2]$, $\beta_{\pm}(x)=\frac{1}{x} \pm \sqrt{ \frac{1}{x(x-1)}\left( \mathcal{P}-\frac{1}{x} \right) }$, $K = \lceil \frac{1}{P} \rceil$.
\end{theorem}

To prove Theorem~\ref{thm:bound}, we first recall the quantum mixed state compiling. Based on the \cite{ezzell2023quantum}, the unique optimal $R$-rank quantum state $\sigma$ to approximate to a $d$-dimensional quantum state $\rho$ such that minimizes the Hilbert-Schmidt distance $D_{HS}(\sigma,\rho) = \Tr[(\sigma-\rho)^2]$ distance is
\begin{align}
  \sigma = \Pi_R\rho\Pi_R + \left(\frac{1-\mathrm{Tr}[\Pi_R\rho\Pi_R]}{R}\right)\Pi_R,
\end{align}
where $\Pi_R$ is a projector onto the eigenstates corresponding to the $R$ largest eigenvalues of $\rho$. In this case, the Hilbert-Schmidt distance reads
\begin{align}\label{eq:DHS_lambda}
  D_{HS} = \sum_{i=R}^{d-1} \lambda^2_{i} + \frac{1}{R^2}\left(\sum_{i=R}^{d-1}\lambda_{i}\right)^2,
\end{align}
where $\lambda_i$ represents the $i$-th eigenvalue of $\rho$ and $1 \ge \lambda_0\ge \lambda_1 \ge \dots \lambda_{d-1} \ge 0$.

Considering the case when $d=3$, we have Lemma~\ref{lemma:min_lambda_i} and Lemma~\ref{lemma:max_lambda_k}.

\begin{lemma}\label{lemma:min_lambda_i}
  The optimal solution of the optimization problem
  \begin{align}
    \min~& \lambda_{i}, \\
    s.t.~& \lambda_{i}^2 +\lambda_{j}^2 + \lambda_{k}^2 = B, \\
    & \lambda_{i} + \lambda_{j} + \lambda_{k} = b, \\
    & b\geq \lambda_{i} \geq \lambda_{j} \geq \lambda_{k} > 0,
  \end{align}
  is equavalent to minimize $\lambda_k$ with the constriants, such that
  \begin{align}
  \lambda_{i} &=  \frac{b}{2} + \sqrt{ \frac{1}{2}\left( B-\frac{b^2}{2} \right) },\label{eq:lambda_i_lk_0}\\
  \lambda_{j} &= \frac{b}{2} - \sqrt{ \frac{1}{2}\left( B-\frac{b^2}{2} \right) }, \label{eq:lambda_j_lk_0}\\
  \lambda_{k} &= 0,
  \end{align}
  when $b^2 \geq B\geq \frac{1}{2}b^2$, and
  \begin{align}
    \lambda_{i}& = \lambda_{j} = \frac{1}{3}b + \frac{1}{6}\sqrt{ 6B - 2b^2 },\label{eq:lambda_i_eq_lambda_j}\\
    \lambda_{k}& = \frac{1}{3}b - \frac{1}{3}\sqrt{ 6B-2b^2 },\label{eq:lambda_k_ne_0}
  \end{align}
  when $\frac{b^2}{2} \geq B \geq \frac{b^2}{3}$. (i.e. $\frac{1}{2}\geq \frac{B}{b^2}\geq \frac{1}{3}$).
\end{lemma}
\begin{proof}
  The proof of Lemma~\ref{lemma:min_lambda_i} is based on the similar spirits of \cite{zhang2024quantification}. The differential of the constraints read
  \begin{align}
    &\lambda_{i} \partial \lambda_{i} + \lambda_{j} \partial \lambda_{j}  + \lambda_{k} \partial \lambda_{k} = 0, \\
    &\partial \lambda_{i} + \partial \lambda_{j}  + \partial \lambda_{k} = 0,
    \end{align}
  which indicate that
  \begin{align}
  \partial\lambda_{i} &= -\frac{\lambda_{k}-\lambda_{j}}{\lambda_{i}-\lambda_{j}} \partial \lambda_{k}, \label{eq:plambdai_k}\\
  \partial\lambda_{j} &= -\frac{\lambda_{i}-\lambda_{k}}{\lambda_{i}-\lambda_{j}} \partial \lambda_{k} \label{eq:plambdaj_k}.
  \end{align}
  Note that $\frac{\partial \lambda_{i}}{\partial \lambda_{k}} =\frac{\lambda_{j}-\lambda_{k}}{\lambda_{i}-\lambda_{j}}\geq {0}$. Hence, minimizing $\lambda_{i}$ requires minimizing $\lambda_{k}$.

  The mean value inequalities suggests that $2(\lambda_j^2+\lambda_k^2)\ge (\lambda_j + \lambda_k)^2$. Then, we have 
  \begin{align}
    2(B-\lambda_k^2)\ge (1-\lambda_k)^2.
  \end{align}
  Then, the lower bound for $\lambda_k$ is 
  \begin{align}
    \lambda_k\ge\max\{0, \frac{1}{3}b - \frac{1}{3}\sqrt{ 6B-2b^2 }\}.
  \end{align}
  Therefore, when $\frac{b^2}{2} \geq B \geq \frac{b^2}{3}$, we have $\lambda_{k} = \frac{1}{3}b - \frac{1}{3}\sqrt{ 6B-2b^2 }$. Then, minimizing $\lambda_i$ implies $\lambda_i = \lambda_j$. Hence, we have \eqref{eq:lambda_i_eq_lambda_j} and \eqref{eq:lambda_k_ne_0}. When $b^2 \geq B\geq \frac{1}{2}b^2$, we have $\lambda_{k} = 0$. By solving the system
  \begin{align}
    \lambda_i^2 + \lambda_j^2 = B,\\
    \lambda_i + \lambda_j = b,
  \end{align}
  we have \eqref{eq:lambda_i_lk_0} and \eqref{eq:lambda_j_lk_0}. 
\end{proof}

\begin{lemma}\label{lemma:max_lambda_k}
  The optimal solution of the optimization problem
  \begin{align}
    \max~ & \lambda_{k} \\
    s.t.~& \lambda_{i}^2 +\lambda_{j}^2 + \lambda_{k}^2 = A \\
    & \lambda_{i} + \lambda_{j} + \lambda_{k} = a \\
    & a\geq \lambda_{i} \geq \lambda_{j} \geq \lambda_{k} > 0
  \end{align}
  is 
\begin{align}
\lambda_{i} = \frac{a}{3} + a\sqrt{ \frac{2}{3}\left( A-\frac{1}{3} \right) }, \\
\lambda_{j}=\lambda_{k} = \frac{a}{3} - \frac{1}{2}\sqrt{ \frac{2}{3}\left( A-\frac{1}{3} \right) }
\end{align}
when $a^2 \geq A \geq \frac{1}{3}a^2$.
\end{lemma}
\begin{proof}
  The proof is obvious by setting $\lambda_j = \lambda_k$.
\end{proof}

Lemma~\ref{lemma:min_lambda_i} and Lemma~\ref{lemma:max_lambda_k} can be generalized into Lemma~\ref{lemma:min_lambda_0_n} and Lemma~\ref{lemma:max_lambdan_n_n}.
\begin{lemma}\label{lemma:min_lambda_0_n}
  The optimal solution of the optimization problem
  \begin{align}
    \min&~ \lambda_{0} \\
    s.t.&~ \sum_{i=0}^{n-1} \lambda_{i}^2 = B \\
    & \sum_{i=0}^{n-1} \lambda_{i} = b \\
    & b\geq \lambda_{0} \geq \lambda_{1} \geq \dots \geq \lambda_{n-1} > 0 \\
  \end{align}
  is
  \begin{align}
    \lambda_{0}=\lambda_{1}=\dots=\lambda_{J-2}=\frac{b-\alpha}{J-1}, \label{eq:lambda_res_0_J_2}\\
    \lambda_{J-1} = \alpha, \label{eq:lambda_res_J_1}\\
    \lambda_{J}=\dots=\lambda_{n-1}=0,\label{eq:lambda_res_J_n_1}
  \end{align}
  where
  \begin{align}
    \alpha = \frac{b}{J} - \sqrt{ \left( 1-\frac{1}{J} \right)\left( B-\frac{b^2}{J} \right) },
  \end{align}
  and $J$ is the integer such that $\frac{b^2}{J} \leq B \leq \frac{b^2}{J-1}$.
\end{lemma}
\begin{proof}
  Suppose we always have the solution in the form as
\begin{align}
  \lambda_{0} = \dots = \lambda_{J -2}, \lambda_{J-1}, \lambda_{J} =\dots = \lambda_{n-1} = 0.
\end{align}
If there exists a group of $\lambda_{0},\dots, \lambda_{n-1}$ that is different from \eqref{eq:lambda_res_0_J_2}, \eqref{eq:lambda_res_J_1}, and \eqref{eq:lambda_res_J_n_1} to obtain smaller $\lambda_{0}$, at least one of the following two cases is satisfied:
\begin{itemize}
  \item There exist $\lambda_{0}, \lambda_{j}, \lambda_{J-1}$, $0 < j < J-1$, $\lambda_{0} > \lambda_{j}$ to obtain smaller $\lambda_{0}$.
  \item There exist $\lambda_{0}, \lambda_{J-1}, \lambda_{l}$, $J-1 < l \le d-1$, $\lambda_{l}>0$ to obtain smaller $\lambda_{0}$.
\end{itemize}

Consider the case that there exist $\lambda_{0}, \lambda_{j}, \lambda_{J-1}$, $0 < j < J-1$, $\lambda_{0} > \lambda_{j}$, such that $\lambda_{0}$ is smaller, while fixing other values and constraints. Then, we can construct an optimization problem as
\begin{align} 
 \min~& \lambda_{0}\\
s.t.~&\lambda_{0}^2 + \lambda_{j}^2 +\lambda_{J-1}^2 = B^\prime, \\
&\lambda_{0} + \lambda_{j} +\lambda_{J-1} = b^\prime, \\
&b^\prime\geq \lambda_{0}\geq \lambda_{j} \geq \lambda_{J-1} \geq 0.
\end{align}
According to Lemma \ref{lemma:min_lambda_i}, we have that $\lambda_{0} = \lambda_{j}$, which reveals contradiction with $\lambda_{0} > \lambda_{j}$.

Then, consider the case of $\lambda_{0}, \lambda_{K-1}, \lambda_{l}$, $K-1 < l \le d-1$, $\lambda_{l}>0$, such that $\lambda_0$ is smaller. Lemma \ref{lemma:min_lambda_i} indicates $\lambda_l = 0$, which reveals contradiction with $\lambda_{l} > 0$.

As a result, \eqref{eq:lambda_res_0_J_2}, \eqref{eq:lambda_res_J_1}, and \eqref{eq:lambda_res_J_n_1} minimizes $\lambda_{0}$.
\end{proof}

\begin{lemma}\label{lemma:max_lambdan_n_n}
  The optimal solution of the optimization problem
  \begin{align}
    \max~& \lambda_{m-1} \\
    s.t.~& \sum_{i=0}^{m-1} \lambda_{i}^2 = A \\
    & \sum_{i=0}^{m-1} \lambda_{i} = a \\
    & a\geq \lambda_{0} \geq \lambda_{1} \geq \dots \geq \lambda_{m-1} \geq 0 \\
  \end{align}
  is
  \begin{align}
    \lambda_{0} = \frac{a}{m} + \sqrt{ \frac{m-1}{m}\left( A - \frac{a^2}{m} \right) }, \label{eq:lambda_0_n_n}\\
    \lambda_{1} = \dots = \lambda_{m-1}= \frac{a}{m} - \frac{1}{m-1}\sqrt{ \frac{m-1}{m}\left( A - \frac{a^2}{m} \right) },\label{eq:lambda_1_m_1_n}
  \end{align}
  when $a^2 \geq A\geq \frac{a^2}{m}$.
\end{lemma}

\begin{proof}
  If there exists a group of $\lambda_{0},\dots, \lambda_{m-1}$ that is different from \eqref{eq:lambda_0_n_n} and \eqref{eq:lambda_1_m_1_n} to obtain larger $\lambda_{m-1}$, it must have $\lambda_1>\lambda_{m-1}$. Consider the subproblem of $\lambda_0$, $\lambda_1$, $\lambda_{m-1}$, $\lambda_1 > \lambda_{m-1}$, while fixing other values and constraints. By Lemma~\ref{lemma:max_lambda_k}, it requires $\lambda_1 = \lambda_{m-1}$ and reveals the contradiction.
\end{proof}

Then, we reform the problem of maximizing $D_{HS}$ as
\begin{align}
  \max~ &D =B + \frac{1}{R^2}b^2, \\
  s.t.~ &A+B = P, ~a + b = 1, \\
  &a = \sum_{i=0}^{R-1} \lambda_{i}, ~b = \sum_{i=R}^{d-1} \lambda_{i}, \\
  &A = \sum_{i=0}^{R-1} \lambda_{i}^2, ~B = \sum_{i=R}^{d-1} \lambda_{i}^2, \\
  &1\geq \lambda_{0} \geq \dots \geq \lambda_{d-1} \geq 0,
\end{align}
where $d-1 > R > 1$.

We first consider the case when $d=3$. Proposition~\ref{proposition:max_D_requires_min_lk} indicates that maximizing $D_{HS}$ is equavalent to minimizing $\lambda_2$ when $R=1$. When $R = 2$, it is obvious that maximizing $D_{HS}$ is equavalent to maximizing $\lambda_2$. 

\begin{proposition}\label{proposition:max_D_requires_min_lk}
  $D_1 = \lambda_j^2 + \lambda_k^2 + u\lambda_j\lambda_k + uv (\lambda_j + \lambda_k)$ is maximized when $\lambda_k$ is minimized with the constraints $\lambda_i + \lambda_j + \lambda_k = b$ and $\lambda_i^2 + \lambda_j^2 + \lambda_k^2 = B$, where $0<u\le 1$ and $0<v<1$.
\end{proposition}

\begin{proof}
  The differential of $D$ reads
  \begin{align}
    \partial D =& (2 \lambda_j + u\lambda_k + uv)\partial \lambda_j + (2 \lambda_k + u\lambda_j + uv)\partial\lambda_k\\
    =&[(2 \lambda_k + u\lambda_j + uv) - (2 \lambda_j + u\lambda_k + uv)\frac{\lambda_{i}-\lambda_{k}}{\lambda_{i}-\lambda_{j}}] \partial\lambda_k \label{eq:introduce_plambdaj}\\
    \le&(2-u) (\lambda_k -\lambda_j) \partial \lambda_k \le 0,\label{eq:ineq_plambdajk}
  \end{align}
  where the equality in \eqref{eq:introduce_plambdaj} introduces \eqref{eq:plambdaj_k}. This indicates that $D$ is non-increasing w.r.t. $\lambda_k$, which means maximizing $D$ requires minimizing $\lambda_k$.
\end{proof}

\begin{lemma}
  $D_{HS}$ in \eqref{eq:DHS_lambda} obtains maximal value when $\max \lambda_{R-1} = \min \lambda_{R}$ such that $\lambda_{1}=\dots=\lambda_{R-1}=\lambda_{R} = \dots = \lambda_{R+J-2}$, where $J = \lceil \frac{B}{b^2} \rceil$, $d-1 > R > 1$.
\end{lemma}

\begin{proof}
  If there exists $\lambda_{i}\ne \lambda_{j}$, $1 \ge i\ge R+J-2$, $1 \ge j\ge R+J-2$ to obtain larger $D_{HS}$, the at least one of the following two cases is satisfied:
\begin{itemize}
  \item $\lambda_{1} > \lambda_{R+J-2}$ that reaches larger $D_{HS}$.
  \item $\lambda_{R+1}>0$ that reaches larger $D_{HS}$.
\end{itemize}

The contradiction of the first case can be revealed by considering the subproblem of $\lambda_0$, $\lambda_{1}$, and $\lambda_{R+J-2}$, while fixing other values and constraints. Note that $D_{HS} = \lambda_{R+J-2}^2 + \frac{1}{R^2}(\sum_{i=R}^{d-1}\lambda_i\lambda_{R+J-2}) + D_c$, where $D_c$ is a constant. Then, we have the derivative $\partial D_{HS}/\partial \lambda_{R+J-2} > 0$, which indicate that $D_{HS}$ is increasing w.r.t. $\lambda_{R+J-2}$. Hence, the subproblem are equavalent to 
\begin{align}
  \max~ & \lambda_{R+J-2}, \\
  s.t.~& \lambda_{0}^2 +\lambda_{1}^2 + \lambda_{R+J-2}^2 = A^\prime, \\
  & \lambda_{i} + \lambda_{j} + \lambda_{k} = a^\prime, \\
  & a\geq \lambda_{i} \geq \lambda_{j} \geq \lambda_{k} > 0.
\end{align}
By Lemma~\ref{lemma:max_lambda_k}, we have $\lambda_{1} = \lambda_{R+J-2}$.

For the second case, we consider the subproblem of $\lambda_{1}$, $\lambda_{R+J-2}$, and $\lambda_{R+J}$, $\lambda_{R+J} > 0$, while fixing other values and constraints. Then we have
\begin{align}
  D_{HS} &= \sum_{i=R}^{d-1} \lambda_{i}^{2}+\frac{1}{R^2}\left( \sum_{i=R}^{d-1} \lambda_{i}\right)^2\\
  &=\left( \frac{R^2+1}{R^2} \right)[\lambda_{R+J-2}^2 + \lambda_{R+J}^2 + \frac{2}{R^2+1} (\lambda_{R+J-2}\lambda_{R+J} + h\lambda_{R+J-2} + h\lambda_{R+J})] + D_{c},
\end{align}
where $h=\sum_{j\in\{R,\dots,d-1\}\setminus\{R+J-2,R+J\}}\lambda_{j}$, $D_{c}$ is a constant. Maximizing $D_{HS}$ is equavalent to maximizing
\begin{align}
  D^\prime_{Hs} = \lambda_{R+J-2}^2 + \lambda_{R+J}^2 + \frac{2}{R^2+1} (\lambda_{R+J-2}\lambda_{R+J} + h\lambda_{R+J-2} + h\lambda_{R+J}).
\end{align}
Note that $0 < \frac{2}{R^2+1} \le 1$, $0< h < 1$. Then, based on the Proposition~\ref{proposition:max_D_requires_min_lk}, maximizing $D^\prime_{HS}$ requires minimizing $\lambda_{R+J}$. Therefore, according to the proof of Lemma~\ref{lemma:min_lambda_i}, we have $\lambda_{R+J} = 0$. The contradiction has been revealed.
\end{proof}

Then we have the theorem about the upper bound of Hilber-Schmidt distance of the quantum-low rank approximation with known purity.

\begin{theorem}\label{thm:bound_quantum_lr_ap}
  Suppose $\mathcal{P}=\Tr (\rho^2)$ is the purity of the quantum state $\rho$, there exist a quantum state $\sigma$ with rank $R$ such that the Hilbert-Schmidt distance
  \begin{align}
    D_{HS}(\sigma,\rho) \le \left\{\begin{aligned}
      &\mathcal{P}-R\beta_{+}^2(K) + \frac{1}{R^2}\left( 1-R \beta_{+}(K)\right)^2, &1 \leq R < K-1,\\
      &\max_{l\ge R}~ (l-R)\left[1+\frac{l-R}{R^2}\right] \beta^2_{-}(l),  &K \leq R < d_{\rho},
    \end{aligned}\right.
  \end{align}
  where $\beta_{\pm}(x)=\frac{1}{x} \pm \sqrt{ \frac{1}{x(x-1)}\left( \mathcal{P}-\frac{1}{x} \right) }$, $K = \lceil \frac{1}{P} \rceil$, and $d_{\rho}$ is the dimension of $\rho$.
\end{theorem}
\begin{proof}
  When $1 \leq R < K-1$, $A/a^2 = 1/R$ can be reached. Suppose we always have the solution in the form as
\begin{align}
  \lambda_{0} = \dots = \lambda_{K -2}, \lambda_{K-1}, \lambda_{K} =\dots = \lambda_{d-1} = 0.
\end{align}
Otherwise, there must exist $\lambda_0$, $\lambda_K-2$, and $\lambda_K-1$, $\lambda_0>\lambda_{K-2}$ and $\lambda_{K-1}\ne 0$ to maximize $D_{HS}$, while fixing other values and constraints. Then, we can construct a subproblem w.r.t. $\lambda_0$, $\lambda_K-2$, and $\lambda_{K-1}$, which indicate that maximizing $D_{HS}$ requires minimizing $\lambda_{K-1}$ based on the Proposition~\ref{proposition:max_D_requires_min_lk}. According to Lemma~\ref{lemma:min_lambda_i}, we have $\lambda_{K-1} = 0$ or $\lambda_0 = \lambda_{K-2}$. Therefore, the contradiction has been revealed. The result of $\lambda_i$, $i=0,\dots,K-1$ can be determined by solving the system
\begin{align}
  (K-1)\lambda_0 + \lambda_{k-1} = 1,\\
  (K-1)\lambda_0^2 + \lambda_{k-1}^2 = \mathcal{P}.
\end{align}
Therefore, we have the optimal solution
\begin{align}
  \lambda_{0}=\lambda_{1}=\dots=\lambda_{K-2}=\beta_{+}(K),\\
  \lambda_{K-1} = 1-(K-1)\beta_{+}(K),\\
  \lambda_{K}=\dots=\lambda_{n-1}=0.
\end{align}
In this case, the Hilbert-Schmidt distance reads
\begin{align}
  D_{HS} = \mathcal{P}-R\beta_{+}^2(K) + \frac{1}{R^2}\left( 1-R \beta_{+}(K)\right)^2.
\end{align}

When $K \leq R < d$, suppose we always have the solution in the form as
\begin{align}
  \lambda_{0}\ge \lambda_{1} = \dots = \lambda_{l-1},
\end{align}
where $\lambda_{l} = 0$ or $l = d$.
Otherwise, there must exist $\lambda_0$, $\lambda_1$, and $\lambda_{l-1}$, $\lambda_1>\lambda_{l-1}$ to maximize $D_{HS}$, while fixing other values and constraints. Then, we can construct a subproblem w.r.t. $\lambda_0$, $\lambda_1$, and $\lambda_{l-1}$, which indicate that maximizing $D_{HS}$ requires maximizing $\lambda_{l-1}$. According to Lemma~\ref{lemma:max_lambda_k}, we have $\lambda_1 = \lambda_{l-1}$. Therefore, the contradiction has been revealed. The result of $\lambda_i$, $i=0,\dots,K-1$ can be determined by solving the system 
\begin{align}
  \lambda_0 + (l-1)\lambda_{l-1} = 1,\\
  \lambda_0^2 + (l-1)\lambda_{l-1}^2 = \mathcal{P}.
\end{align}
Therefore, we have the optimal solution
\begin{align}
  \lambda_{0} = 1-(l-1)\beta_{-}(l),\\
  \lambda_{1}=\dots=\lambda_{l-1}=\beta_{-}(l),\\
  \lambda_{l}=\dots=\lambda_{d-1}=0.
\end{align}
In this case, the Hilbert-Schmidt distance reads
\begin{align}
  D_{HS} = (l-R)\left[1+\frac{l-R}{R^2}\right] \beta^2_{-}(l).
\end{align}
\end{proof}

From \cite{bisio2011Minimal}, the minimal dimension of ancillary space $\mathcal{H}_{A_N}$ is the size of $\mathrm{Supp}(\Upsilon)$. In other word, the there exist a list of isometries $d_{A_1} =\dots = d_{A_N} = R$ that implement the Choi state $\Upsilon^\prime$ of the quantum comb $\mathcal{C}$, where $\mathrm{rank}(\Upsilon^\prime) = R$. In this circumstance, the number of entries of isometries is $O(\max_{0\leq t< N}(d_{\mathsf{i}_{t}}d_{\mathsf{o_{t}}}) R^2)$. Based on the Theorem~\ref{thm:bound_quantum_lr_ap}, considering the normalizing factor, we have that
\begin{align}
  D_{HS}(\Upsilon^\prime,\Upsilon) \le \left\{\begin{aligned}
    &(\Tr\Upsilon)^2[\mathcal{P}-R\beta_{+}^2(K) + \frac{1}{R^2}\left( 1-R \beta_{+}(K)\right)^2], &1 \leq R < K-1,\\
    &\max_{l\ge R} (\Tr\Upsilon)^2(l-R)\left[1+\frac{l-R}{R^2}\right] \beta^2_{-}(l)  &K \leq R < d.
  \end{aligned}\right.
\end{align}
The Theorem~\ref{thm:bound} has been proved.

\begin{figure}[t]
  \includegraphics[width=0.8\textwidth]{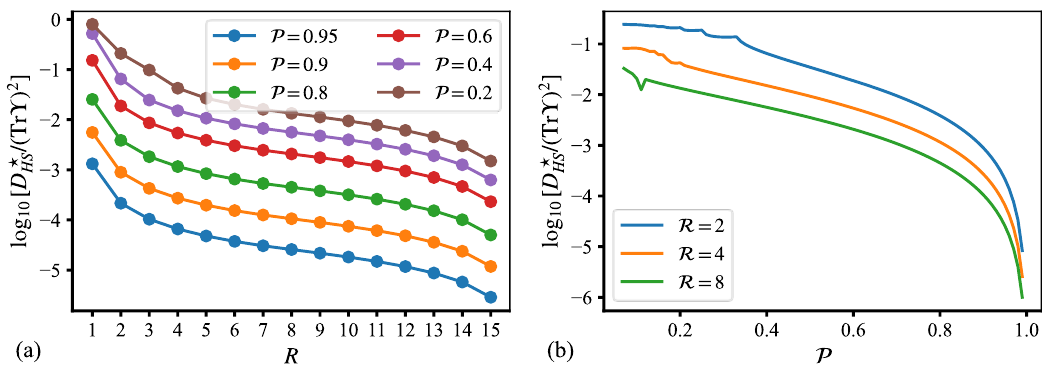}
  \caption{\label{fig:purity_dist} Upperbound of Hilbert-Schmidt distance in Theorem~\ref{thm:bound} to implement a 16-dimensional Choi state.}
\end{figure}

We showcase the upperbound of the Hilbert-Schmidt distance $D_{HS}^{\star}$ in Fig.~\ref{fig:purity_dist}. The dimension of target Choi state of the quantum comb is set as $16$. Note that $D_{HS}^\star$ is not smooth and non-increasing w.r.t. the purity while R is specified. This phenomenon may results from the non-smooth eigen values and the transition of the existance of the real square root in $\beta_{\pm}$.

\section*{Tangent Space, Velocity, and Retractions on the Stiefel manifold}

Stiefel manifold $\mathrm{St}(n,p)$ is an embedded submanifold of $\mathbb{C}^{n\times p}$ ($\mathbb{R}^{n\times p}$ for real cases), which is defined as the set of $p$ orthonormal vectors in $\mathbb{C}^{n}$ \cite{boumal2023Introduction}. An element in $\mathrm{St}(n,p)$ can be represented as a complex matrix $X\in\mathbb{C}^{n\times p}$ such that $X^{\dagger} X = I$.

To optimize the cost function defined in (4) in the main text with the orthonormal constraints, we want to
\begin{itemize}
  \item initialize the point $X$ on the Stiefel manifold,
  \item for each point $X$ on the Stiefel manifold, find a velocity $V$ that decreases the cost function,
  \item and move the point along $V$ on the Stiefel manifold.
\end{itemize}
The first tip can be trivially implemented by randomly choosing an orthonormal matrix. The second and third tips require that find a direction on the tangent space of $\mathrm{St}(n,p)$ at $X$, and choose a retraction $R$ to move smoothly on the Stiefel manifold such that $R({0})\vert_X = X$ and $R^\prime({0})\vert_X = V$.

The tangent space $\mathcal{T}_X \mathcal{M}$ of a manifold $\mathcal{M}$ is the linear space of derivatives of all smooth curves $R(t)$ on the manifold at point $X$,
\begin{align}
  \mathcal{T}_X \mathcal{M} = \{R^\prime(0)\vert R: \mathcal{I} \to \mathcal{M} \mbox{ is smooth and } R(0) = X\},
\end{align}
where $\mathcal{I}$ is any open interval containing $t=0$. That is, $Z$ is in $\mathcal{T}_X \mathcal{M}$ if and only if there exists a smooth curve on $\mathcal{M}$ passing through $X$ with velocity $Z$. Hence, the tangent space of the Stiefel manifold $\mathrm{St}(n,p)$ at $X\in\mathrm{St}(n,p)$ can be described as 
\begin{align}
  \mathcal{T}_X \mathrm{St}(n,p) = \{Z \vert Z^\dagger X + X^\dagger Z = 0\}.
\end{align}

The velocity $V$ at $X\in\mathrm{St}(n,p)$ should be projected onto the tangent space $\mathcal{T}_X \mathrm{St}(n,p)$ to move the point along $V$ on the Stiefel manifold. Based on the canonical inner product \cite{tagare2011Notes},
\begin{align}\label{eq:canonical_innerprod}
  \langle Z_1, Z_2 \rangle_c = \Tr[Z_1^\dagger(I-\frac{1}{2}XX^\dagger)Z_2],
\end{align}
the velocity $V$ is projected as $U = DX$, where 
\begin{align}
  D = VX^\dagger - XV^\dagger.
\end{align}
An example of the velocity projection is the Riemannian gradient of $f:\mathrm{St}(n,p)\to\mathbb{R}$ at $X$, 
\begin{align}
  U = G - X G^\dagger X,
\end{align}
where $G = \frac{\partial f}{\partial X^*}$.

To move the point along $V$ on the Stiefel manifold, a retraction $R$ that continuously wraps the tangent space to the manifold using a curve $C(X,V)$ on the manifold,
\begin{align}
  R:\mathcal{T}_X \mathrm{St}(n,p)\to \mathrm{St}(n,p): (X,V) \to C(X,V),
\end{align}
where $R(t) = C(X, tV)$ that satisfies $R(0) = X$, and $R^\prime(0) = V$. In this letter, we utilizes the Cayley retraction that 
\begin{align}
  R(t) = C(X,D) = (I+t\frac{D}{2})^{-1}(I-t\frac{D}{2})X,
\end{align}
where $U = DX$ is the projection of the velocity $V$ on the tangent space $\mathcal{T}_X \mathrm{St}(n,p)$.

\section*{Solving unconstrained problem on Stiefel manifold via ADAM}
Adaptive moment estimation (ADAM) is a prevalent first-order optimization method for stochastic differentiable scalar cost functions $f(\bm{\theta})$ w.r.t. parameters $\bm{\theta}$ \cite{kingma2014Adam}. The main goal of ADAM is to minimize the expected value of the cost function $\mathbb{E}[f(\bm{\theta})]$. The algorithm estimates the first and second moment $\bm{m}$ and $v$ as the exponential moving averages of the gradient and squared gradient, respectively, where hyper-parameters $\gamma_1, \gamma_2\in[0,1)$ control the exponential decay rates of these moments.

At the beginning of ADAM, moments $\bm{m}$ and $v$ are initialized as $\bm{0}$ and $1$, respectively, which leads to the biased moment estimation. Therefore, bias correction is required at each iteration,
\begin{align}\label{eq:bias_correction}
  \hat{\bm{m}} = \frac{\bm{m}}{1-\gamma_1^t},~\hat{v} =\frac{v}{1-\gamma_2^t} ,
\end{align}
where $t$ is the current iteration. Based on the bias-corrected moments, parameters are updated by
\begin{align}
  \bm{\theta} = \bm{\theta} - \kappa \frac{\hat{\bm{m}}}{\sqrt{\hat{v}}},
\end{align}
where $\kappa = \min(\kappa_0, 1/\|\frac{\hat{\bm{m}}}{\sqrt{\hat{v}}}\|)$ is the adaptive learning rate, and $\kappa_0$ is a hyperparameter that limits the maximum learning rate. 

Moving onto the Stiefel manifold, we consider the cost function $f: \mathrm{St} \to \mathbb{R}$ w.r.t. $X\in \mathrm{St}$. At iteration $t$, we update the biased first moment $M$ and second moment $v$ by
\begin{align}
  M &\leftarrow \gamma_1 M + (1-\gamma_1) \frac{\partial f}{\partial X^*},\\
  v &\leftarrow \gamma_2 v + (1 - \gamma_2)\|\frac{\partial f}{\partial X^*}\|_{F}^2.
\end{align}
Then, we compute the bias-corrected first moment $\hat{M}$ and second moment $\hat{v}$ by \eqref{eq:bias_correction}. Finally, the parameter matrix $X$ is updated by the Cayley retraction with the velocity determined by projecting bias-corrected moments onto $\mathcal{T}_X \mathrm{St}$,
\begin{align}
  X \leftarrow (I+\kappa\frac{{D}}{2})^{-1}(I-\kappa\frac{{D}}{2})X,
\end{align}
where 
\begin{align}
  D = \frac{1}{\sqrt{\hat{v}}}(\hat{M}X^\dagger - X\hat{M}^\dagger),
\end{align}
and $\kappa=\min\{\kappa_0,1/(\|D\|_F+\epsilon)\}$.

We summarize the Stiefel ADAM in Algorithm~\ref{alg:adam_stiefel}, where $\gamma_1$, $\gamma_2$, $\kappa_0$, $\delta$, and $\epsilon$ are hyperparameters. Note that the hyperparameter $\epsilon > 0$ is introduced to avoid the $0$ denominator. In this letter, we empirically set $\gamma_1 = 0.9$, $\gamma_2 = 0.999$, and $\epsilon=10^{-8}$, while $\kappa_0$ and $\delta$ are set that varies from experiments.

\begin{algorithm}[h]
  \caption{ADAM on the Stiefel manifold}\label{alg:adam_stiefel}
  \KwIn{Initial parameter $X_{0}$ on the Stiefel manifold}
  $X \leftarrow X_{0}$, $M \leftarrow \bm{0}$, $v \leftarrow 1$; \tcp*[f]{Initialize parameter and first and second moments}\\
  \For{$t = 1$ \KwTo $T$}{
    $M \leftarrow \gamma_1 M + (1-\gamma_1) \frac{\partial f}{\partial X^*}$;\tcp*[f]{Update biased first moment}\\
    $v \leftarrow \gamma_2 v + (1 - \gamma_2)\|\frac{\partial f}{\partial X^*}\|_{F}^2$; \tcp*[f]{Update biased second moment}\\
    $r\leftarrow(1-\gamma_1^t)\sqrt{v_t/(1-\gamma_2^t)+\epsilon}$;\tcp*[f]{Estimate biased-corrected ratio}\\
    $D \leftarrow \frac{1}{r}({M}X^\dagger - X{M}^\dagger)$;\tcp*[f]{Project onto the tangent space}\\ 
    $\kappa\leftarrow\min\{\kappa_0,1/(\|D\|_F+\epsilon)\}$ ;\tcp*[f]{Select adaptive learning rate}\\
    $X \leftarrow (I+\kappa\frac{D}{2})^{-1}(I-\kappa\frac{D}{2})X$;\tcp*[f]{Update $X$ by Cayley retraction}\\
    \If{$\|G X^\dagger - XG^\dagger\|_F<\delta$}{
      $\mathbf{break}$\\
    }
  }
  \Return{$X$.}
  
\end{algorithm}

\section*{Isometry based Quantum Comb Tomography}

\begin{algorithm}[h]
  \caption{\label{al:iqct}Framework of iQCT}
  \KwIn{Complete sets $\{\Gamma^{(k)}\}_{k=0}^{N-1}$ and $\{\Xi^{(k)}\}_{k=0}^{N-1}$, and experiment data $\{s^{(k)}_{\bm{\alpha},\bm{\beta}}\}_{k,\bm{\alpha},\bm{\beta}}$}
  \Begin{
    $\{\eta^{(-1)}_{\bm{\alpha}, \bm{\beta}}\}\leftarrow\{\rho^{(0)}_{\alpha_0}\}$; \tcp*[f]{Initialize temporary states}\\
    $V^{(0)}\leftarrow \arg\min_{W^{(0)}\in\mathrm{St}^{(0)}}\mathcal{F}(W^{(0)})$;\\
    $\mathbb{V}\leftarrow \{V^{(0)}\}$;\tcp*[f]{Result set of isometries}\\
    \For{$k = 1, \dots, N-1$}{
      \{$\eta^{(k-1)}_{\bm{\alpha}, \bm{\beta}}\}\leftarrow \mathbf{GetTempState}(\{\Tr_{\mathtt{o}_{k-1}} [\rho^{(k)}_{\alpha_{k}} E^{{(k-1)}}_{\beta_t}V^{(k-1)}\eta^{(k-2)}_{\bm{\alpha},\bm{\beta}}V^{(k-1)\dagger}]\})$;\\
      $V^{(k)}\leftarrow \arg\min_{W^{(k)}\in\mathrm{St}^{(k)}}\mathcal{F}(W^{(k)})$;\\
      $\mathbb{V}\leftarrow \mathbb{V} \cup \{V^{(k)}\}$;\\
    }
  }
  \Return{$\mathbb{V}$.}
\end{algorithm}

The isometry based quantum comb tomography determines isometries $[V^{(0)}, \dots, V^{(N-1)}]$ that completely represent the target $N$-time-step quantum comb. The experimenter prepares
\begin{itemize}
  \item[(1)] complete sets of input states $\Gamma^{(k)}:=\{\rho^{(k)}_{i}\in\mathrm{Lin}(\mathcal{H}_{\mathtt{i}_k})\}_{i=0}^{d_{2k}^2-1}$ for input systems, $k = 0,1,\dots,N-1$,
  \item[(2)] complete sets of measurements $\Xi^{(k)}:=\{E^{(k)}_{j} \in \mathrm{Lin}(\mathcal{H}_{\mathtt{o}_k})\}_{j=0}^{d_{2k+1}^2-1}$ for output systems, $k = 0,1,\dots,N-1$,
\end{itemize}
where each complete set of input states $\Gamma^{(k)}$ consists of $d_{\mathtt{i}_k}^2$ linear independent quantum states, and each complete set of measurements $\Xi^{(k)}$ consists of $d_{\mathtt{o}_k}^2$ linear independent POVM operators. These input states and measurements are mathematically known to the experimenter.

Then, the stepwise optimization is performed to determine isometries. At time step $k$, only isometry $V^{(k)}$ is determined with known $V^{(t)}$, $t<k$. Let $\bm{\alpha} := [\alpha_0, \dots, \alpha_k]$ and $\bm{\beta} := [\beta_0, \dots, \beta_k]$ be two lists of indexes representing which input states and measurements are utilized, respectively. Specifically, input states and measurements of an experiment labeled by $\bm{\alpha}$ and $\bm{\beta}$ are $\{\rho^{(0)}_{\alpha_0},\dots,\rho^{(k)}_{\alpha_k}\}$ and $\{E^{(0)}_{\beta_0},\dots,E^{(k)}_{\beta_k}\}$, respectively. The experimenter conducts experiments labeled by $\bm{\alpha}$ and $\bm{\beta}$ and records the results $s^{(k)}_{\bm{\alpha},\bm{\beta}}$. The criteria for selecting $\bm{\alpha}$ and $\bm{\beta}$ are that $\{\eta^{(k-1)}_{\bm{\alpha},\bm{\beta}}\}$ should consist of at least $d^2_{\mathtt{i}_k}d^2_{A_k}$ linear independent matrices and that $\beta_k$ spans $\{0,\dots,d^2_{\mathtt{o}_k} -1\}$, where $\eta^{(k-1)}_{\bm{\alpha},\bm{\beta}}$ is a temporary state that is determined recursively,
\begin{gather}
  \eta^{(t)}_{\bm{\alpha},\bm{\beta}} = \Tr_{\mathtt{o}_t} [\rho^{(t+1)}_{\alpha_{t+1}} E^{{(t)}}_{\beta_t}V^{(t)}\eta^{(t-1)}_{\bm{\alpha},\bm{\beta}}V^{(t)\dagger}], t \ge 0,
\end{gather}
and $\eta^{(-1)}_{\bm{\alpha},\bm{\beta}} = \rho^{(0)}_{\alpha_0}$. The minimum set of required temporary states can be easily implemented, since we have $\{\eta_{\bm{\alpha}^\prime, \bm{\beta}^{\prime}}^{(k-2)}\}$ at time step $k-1$ to determine $V^{(k-1)}$, by finding the maximum linear independent subset of 
\begin{align}
  \{\Tr_{\mathtt{o}_{k-1}} [\rho^{(k)}_{\alpha_{k}} E^{{(k-1)}}_{\beta_{k-1}}V^{(k-1)}\eta_{\bm{\alpha}^\prime, \bm{\beta}^{\prime}}^{(k-2)}V^{(k-1)\dagger}]\}_{\alpha_k=0,\beta_{k-1}=1}^{d^2_{\mathtt{i}_{k-1}}-1, d^2_{\mathtt{o}_{k-1}}-1}.
\end{align}
Note that the recursive determination of the temporary state only requires matrix multiplications with maximum dimensions $d_{\mathtt{i}_{t+1}}d_{\mathtt{i}_{t}}d_{A_{t}}$. The determined temporary states can be stored into memory to avoid repetitive computation. 

Then, the cost function is composed as
\begin{align}
  \mathcal{F}(W^{(k)}) = \sum_{\bm{\alpha},\bm{\beta}} \vert \tilde{p}_{\bm{\alpha},\bm{\beta}} - \Tr[E^{(k)}_{\beta_{k}}W^{(k)} \eta^{(k-1)}_{\bm{\alpha},\bm{\beta}} W^{(k)\dagger}]\vert^2,
\end{align}
where $\tilde{p}_{\bm{\alpha},\bm{\beta}}=s^{(k)}_{\bm{\alpha},\bm{\beta}}/n_s$ is the measurement probability. Hence, the first derivative of the cost function can be determined by
\begin{align}
  \frac{\partial \mathcal{F}}{\partial W^{(k)*}} = \sum_{\bm{\alpha},\bm{\beta}}2\left(\Tr[E^{(k)}_{\beta_{k}}W^{(k)} \eta^{(k-1)}_{\bm{\alpha},\bm{\beta}} W^{(k)\dagger}] - \tilde{p}_{\bm{\alpha},\bm{\beta}}\right)E^{(k)}_{\beta_{k}}W^{(k)}\eta^{(k-1)}_{\bm{\alpha},\bm{\beta}}.
\end{align}

Then, the isometry $V^{(k)}$ can be reconstructed by optimizing $\mathcal{F}$ with orthonormal constraints
\begin{align}\label{eq:euc_problem}
  \min_{W^{(k)}\in\mathbb{C}^{(k)}}&\mathcal{F}(W^{(k)}),\\
  s.t.~~~ &W^{(k)^\dagger}W^{(k)} = I,
\end{align}
where $\mathbb{C}^{(k)}:=\mathbb{C}^{d_{\mathtt{o}_k}d_{A_{k+1}}\times d_{\mathtt{i}_k}d_{A_k}}$. The constraint $W^{(k)^\dagger}W^{(k)} = I$ indicates that $W^{(k)}$ lies on the Stiefel manifold $\mathrm{St}^{(k)}:=\{X\in \mathbb{C}^{(k)}:X^\dagger X = I\}$. Hence, the orthonormally constrained optimization problem is transformed into an unconstrained problem on the Stiefel manifold 
\begin{align}\label{eq:stiefel_problem}
  \min_{W^{(k)}\in\mathrm{St}^{(k)}}\mathcal{F}(W^{(k)}).
\end{align}
Then, we have $V^{(k)} = \arg\min_{W^{(k)}\in\mathrm{St}^{(k)}}\mathcal{F}(W^{(k)})$ based on the ADAM on the Stiefel manifold. The framework of iQCT is summarized in Algorithm~\ref{al:iqct}, where $\mathbf{GetTempState}$ outputs a required set of temporary states specified by the experimenter.

\section*{Instrument-based iQCT}
The instrument-based iQCT requires a complete set of initial states $\Gamma^{(0)}:=\{\rho^{(0)}_{\alpha}\}_{\alpha=0}^{d^2-1}$ and complete sets of CPTNI instruments
\begin{align}
  \mathcal{J}^{(t)} = \{\mathcal{A}_{x_t}^{(t)}\}_{x_t=0}^{d^4-1},~ t=0,1,\dots, N-1,
\end{align}
where $\mathcal{A}_{x_t}^{(t)}$ transforms the output state from $\mathrm{Lin}(\mathcal{H}_{\mathtt{o}_t})$ to the input state $\mathrm{Lin}(\mathcal{H}_{\mathtt{i}_{t+1}})$. As shown in Fig.~\ref{fig:inst_qct}, the experimenter prepares the initial state $\rho^{(0)}$ and performs instruments in each time slots to obtain measurement data.

\begin{figure}[t]
  \includegraphics[width=0.79\textwidth]{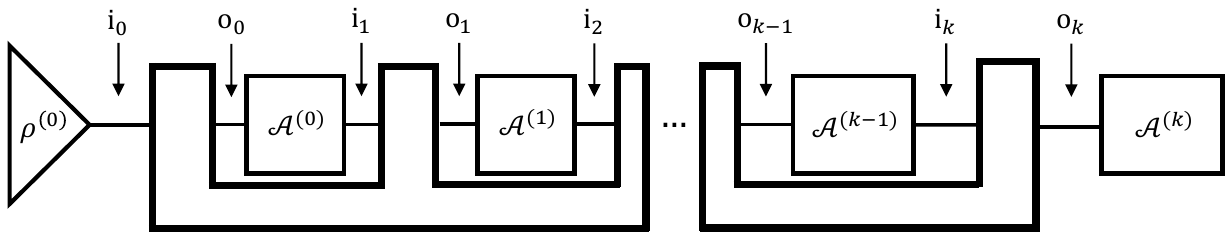}
  \caption{\label{fig:inst_qct} Experiment framework of instrument-based iQCT.}
\end{figure}

The iQCT still applies the stepwise optimization to learning isometries. At step $k$, the isometry $V^{(k)}$ is recovered with known $V^{(t)}$, for $t<k$. The experimenter conducts experiments by performing interventions of a sequence of instruments $[\mathcal{A}_{x_0}^{(0)},\dots,\mathcal{A}_{x_{k}}^{(k)}]$ indexed by $\bm{x}:=[x_0,\dots, x_k]$ on the initial states $\rho_{\alpha}^{(0)}$, and obtains the data $s_{\alpha,\bm{x}}^{(k)}$. The criteria to select $\alpha$ and $\bm{x}$ are that $\{\eta^{(k-1)}_{\alpha,\bm{x}}\}$ consists of at least $d^2d^2_{A_k}$ linear independent matrices, where
\begin{gather}
  \eta^{(t)}_{\alpha,\bm{x}} = \Tr_{\mathtt{o}_t} [\xi_{x_{t}}^{(t)} V^{(t)}\eta^{(t-1)}_{\alpha,\bm{x}}V^{(t)\dagger}], t \ge 0,
\end{gather}
and $\eta^{(-1)}_{\bm{\alpha},\bm{\beta}} = \rho^{(0)}_{\alpha}$, where $\xi_{x_{t}}^{(t)} := \sum_{ij}[\mathcal{A}_{x_t}^{(t)}(\vert i\rangle \langle j\vert)]_{\mathtt{i}_{t+1}}\otimes [\vert j\rangle \langle i\vert]_{\mathtt{o}_{t}}$ such that $\Tr_{\mathtt{o}_t}[\xi_{x_{t}}^{(t)}(I_{\mathtt{i}_{t+1}}\otimes \rho_{\mathtt{o}_t})] = \mathcal{A}_{x_t}^{(t)}(\rho_{\mathtt{o}_t})$. Hence, the recovered probability of the experiemnt is
\begin{gather}
  p_{{\alpha,\bm{x}}} (W^{(k)})=\Tr[\xi_{x_{k}}^{(k)} W^{(k)}\eta^{(k-1)}_{\alpha,\bm{x}}W^{(k)\dagger}].
\end{gather} 
Further processes are the same as regular iQCT.

\section*{Characterizing Non-Markovian Quantum Noise}

A typical quantum device with a $d$-dimension system under non-Markovian quantum noise executes tasks by performing a sequence of completely positive and trace non-increasing (CPTNI) instruments on the initial state. Since the non-Markovian quantum noise can be modeled by the quantum comb, the iQCT is feasible to characterize non-Markovian quantum noise on quantum devices.

However, the instrument-based should be modified before characterizing non-Markovian quantum noise on the noisy intermediate-scale quantum (NISQ) devices due to the their limits. On a typical $d$-dimensional NISQ device, intermediate instruments are completely positive and trace-preserving (CPTP) operations, while the instrument at the last time step is a POVM measurement. Therefore, the instrument-based iQCT requires a complete set of initial states $\Gamma^{(0)}$, complete POVM measurements $\Xi^{(k)}$, $k=0,\dots,N-1$ and complete sets of CPTP operations $\mathcal{J}^{(t)}$, $t=0,\dots,N-2$.

At step $k$, the isometry $V^{(k)}$ is recovered with known $V^{(t)}$, for $t<k$. The experimenter conducts experiments by performing interventions of a sequence of instruments $[\mathcal{A}_{x_0}^{(0)},\dots,\mathcal{A}_{x_{k-1}}^{(k-1)}]$ indexed by $\bm{x}:=[x_0,\dots, x_{k-1}]$ on the initial states $\rho_{\alpha}^{(0)}$. Then, data $s_{\alpha,\bm{x},\beta}^{(k)}$ are measured by performing measurement $E^{(k)}_{\beta}$ on the final output states. The criteria to select $\alpha$, $\bm{x}$, and $\beta$ are that $\{\eta^{(k-1)}_{\alpha,\bm{x}}\}$ consists of at least $d^2d^2_{A_k}$ linear independent matrices and that $\beta$ spans $\{0,\dots,d^2-1\}$, where
\begin{gather}
  \eta^{(t)}_{\alpha,\bm{x}} = \Tr_{\mathtt{o}_t} [\xi_{x_{t}}^{(t)} V^{(t)}\eta^{(t-1)}_{\alpha,\bm{x}}V^{(t)\dagger}], t \ge 0,
\end{gather}
and $\eta^{(-1)}_{\bm{\alpha},\bm{\beta}} = \rho^{(0)}_{\alpha}$. Hence, the recovered probability of the experiemnt is
\begin{gather}
  p_{{\alpha,\bm{x},\beta}} (W^{(k)})=\Tr[E^{(k)}_{\beta}W^{(k)}\eta^{(k-1)}_{\alpha,\bm{x}}W^{(k)\dagger}].
\end{gather} 
Further processes are the same as regular iQCT.

\begin{figure}[t]
  \includegraphics[width=0.79\textwidth]{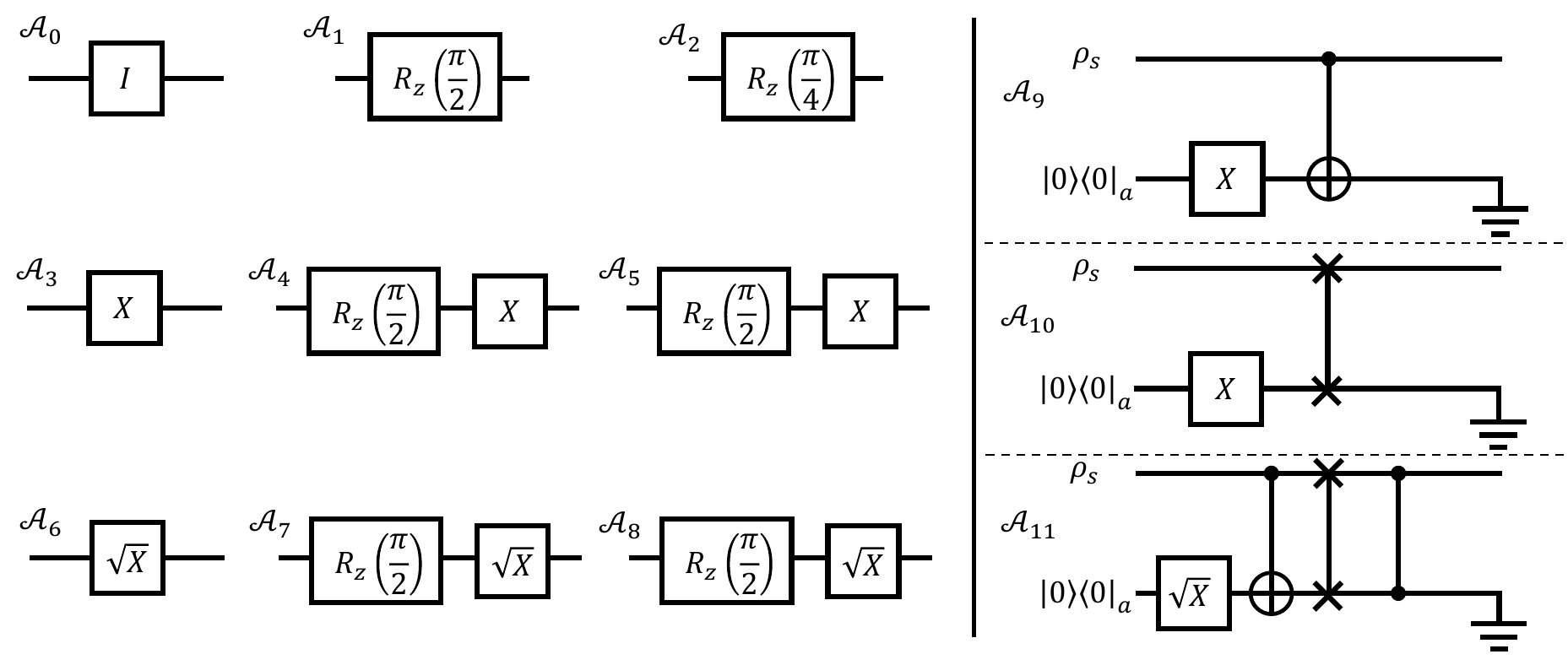}
  \caption{\label{fig:cptp_inst} Single-qubit CPTP instruments for NISQ devices. $\mathcal{A}_0,...,\mathcal{A}_{11}$ span the space of single-qubit CPTP operations, where $\mathcal{A}_0,...,\mathcal{A}_8$ span the unitary space, $\mathcal{A}_9,...,\mathcal{A}_{11}$ are CPTP operations implemented with ancilla qubits with subscript $a$, and $\rho_s$ represents the quantum system.}
\end{figure}

Both simulations and real-device experiments are conducted in the NISQ environment. For simulations, we utilize $\mathcal{A}_0,\dots,\mathcal{A}_{11}$ at intermediate time steps. Initial states are defined as
\begin{align}
  \rho^{(0)}_{0} = 
  \begin{bmatrix}
    0.5& -0.5\\
    -0.5&  0.5
  \end{bmatrix},~\rho^{(0)}_{1} = 
  \begin{bmatrix}
    0.5&  0.5\\
    0.5&  0.5
  \end{bmatrix},~\rho^{(0)}_{2} = 
  \begin{bmatrix}
    0.5&  -0.5i\\
    0.5i&  0.5
  \end{bmatrix},~\rho^{(0)}_{3} = 
  \begin{bmatrix}
    1&  0\\
    0&  0
  \end{bmatrix}.
\end{align} 
POVM operators are similarly defined as
\begin{align}
  E^{(k)}_{0} = 
  \begin{bmatrix}
    0.5& -0.5\\
    -0.5&  0.5
  \end{bmatrix},~E^{(k)}_{1} = 
  \begin{bmatrix}
    0.5&  0.5\\
    0.5&  0.5
  \end{bmatrix},~E^{(k)}_{2} = 
  \begin{bmatrix}
    0.5&  -0.5i\\
    0.5i&  0.5
  \end{bmatrix},~E^{(k)}_{3} = 
  \begin{bmatrix}
    1&  0\\
    0&  0
  \end{bmatrix}.
\end{align}
Simulation results demonstrate similar properties shown by Fig.~3(a) in the main text.

We also apply the iQCT to a single-qubit real superconductive quantum device produced by Yangtze Delta Region Industrial Innovation Center of Quantum and Information Technology, Suzhou. The maximum time step is set as $N=3$. Due to the restriction of the device, we can only guarantee the non-measurement operations that span the unitary space. We utilize $\mathcal{A}_0,\dots,\mathcal{A}_{8}$ at intermediate time steps. Initial states are prepared as
\begin{align}
  \rho^{(0)}_{0} = 
  \begin{bmatrix}
    1& 0\\
    0&  0
  \end{bmatrix},~\rho^{(0)}_{1} = 
  \begin{bmatrix}
    0&  0\\
    0&  1
  \end{bmatrix},~\rho^{(0)}_{2} = 
  \begin{bmatrix}
    0.5&  0.5\\
    0.5&  0.5
  \end{bmatrix},~\rho^{(0)}_{3} = 
  \begin{bmatrix}
    0.5&  -0.5i\\
    0.5i&  0.5
  \end{bmatrix}.
\end{align} 
POVM operators are performed as 
\begin{align}
  E^{(k)}_{0} = 
  \begin{bmatrix}
    1& 0\\
    0&  0
  \end{bmatrix},~E^{(k)}_{1} = 
  \begin{bmatrix}
    0&  0\\
    0&  1
  \end{bmatrix},~E^{(k)}_{2} = 
  \begin{bmatrix}
    0.5&  -0.5i\\
    0.5i&  0.5
  \end{bmatrix},~E^{(k)}_{3} = 
  \begin{bmatrix}
    0.5&  -0.5\\
    -0.5&  0.5
  \end{bmatrix}.
\end{align}
An experiment is conducted as shown in Fig.~\ref{fig:rc_exp_frame}, where the operations are conducted within the same time gap $\tau = 10$ns. Before running the iQCT, we first correct the unitaries and POVM operators by performing gate set tomography in the unitary space. We estimate the non-Markovian quantum noise with full ancillary dimensions as $\mathcal{C}_{\mathrm{full}}$, where $d_{A_1} = 4$, $d_{A_2} = 16$, and $d_{A_3} = 32$. Then, we specify the ancillary dimensions as $[d_{A_1},d_{A_2},d_{A_3}] = [2,2,2]$, $[2,2,4]$, and $[2,4,8]$, respectively, to determine dimension-reduced results $\mathcal{C}_{\mathrm{dr}}$. 

\begin{figure}[h]
  \includegraphics[width=0.79\textwidth]{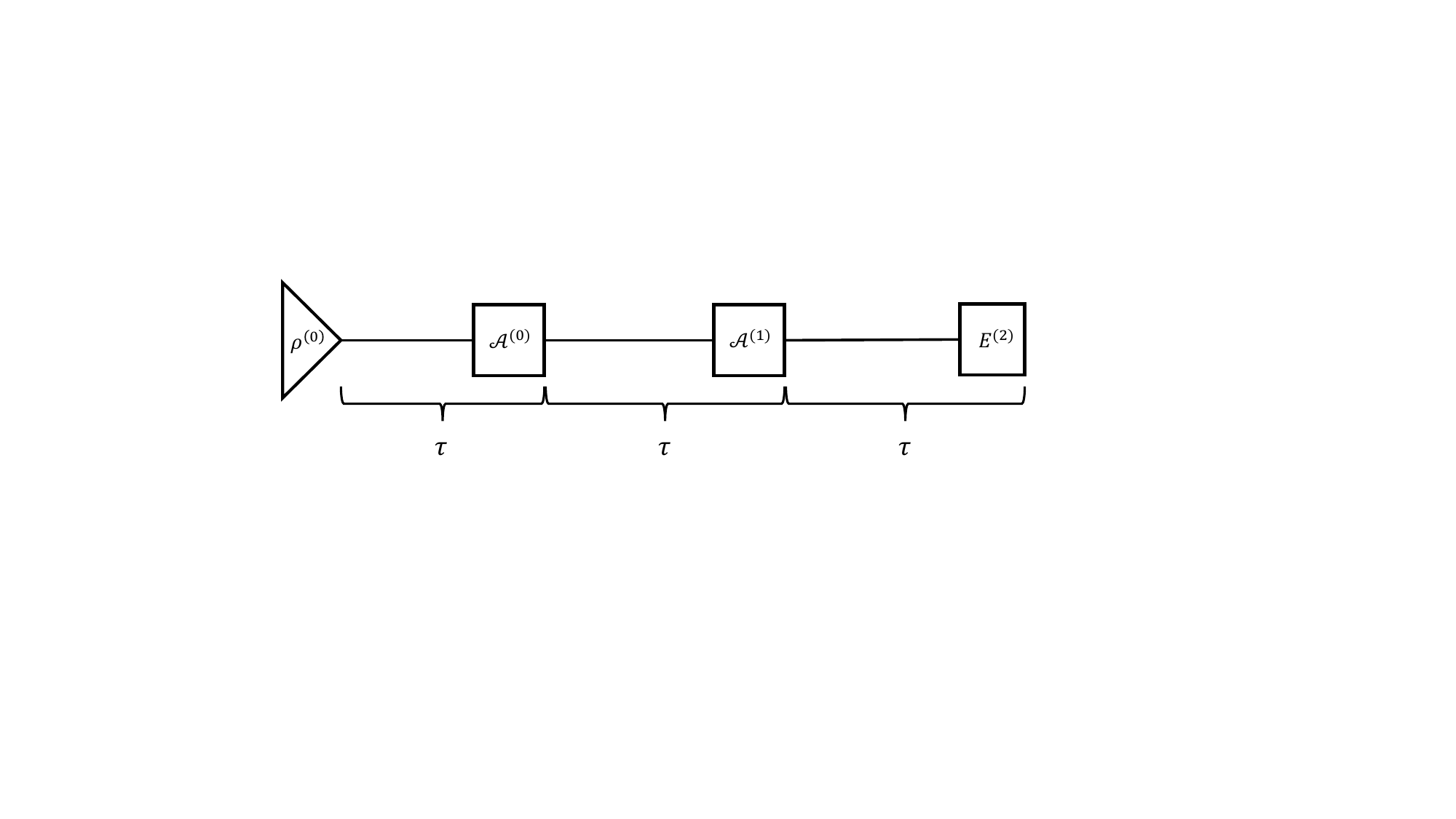}
  \caption{\label{fig:rc_exp_frame} Experiment on the real quantum device. Operations are conducted within the same time gap $\tau$.}
\end{figure}

We utilize the relative cost $\mathcal{L}(\mathcal{C}_{\mathrm{dr}},\mathcal{C}_{\mathrm{full}})$ to measure distances between $\mathcal{C}_{\mathrm{dr}}$ and $\mathcal{C}_{\mathrm{full}}$,
\begin{align}
  \mathcal{L}(\mathcal{C}_{\mathrm{dr}},\mathcal{C}_{\mathrm{full}}) = \sum_{{\alpha,\bm{x},\beta}} \vert p_{{\alpha,\bm{x},\beta}}(\mathcal{C}_{\mathrm{dr}}) - p_{{\alpha,\bm{x},\beta}}(\mathcal{C}_{\mathrm{full}})\vert^2,
\end{align}
where $p_{{\alpha,\bm{x},\beta}}(\bullet)$ is the recovered probability of the operand quantum comb, $\mathcal{A}^{(k)}_{i}\in\{\mathcal{A}_0,\dots,\mathcal{A}_{8}\}$ span the unitary space for all $k$. The relative cost measures the distance between two operand quantum combs where non-measurement operations are defined in the unitary spaces.

\bibliography{citations}